\documentclass[11pt, letterpaper]{article}

\usepackage[margin=1in]{geometry}
\usepackage[left]{lineno}

\usepackage[T1]{fontenc}
\usepackage{lmodern}
\usepackage{amsmath, amssymb, amsthm, dsfont, mathtools}
\usepackage{thmtools}
\usepackage{thm-restate}
\usepackage{xspace}
\usepackage{graphicx}
\usepackage[dvipsnames]{xcolor}
\usepackage{tikz}
\usetikzlibrary{decorations.pathmorphing, arrows.meta}
\usepackage{subcaption}
\usepackage{nicefrac}
\usepackage{bm}
\usepackage{thm-restate}
\usepackage{enumitem}
\usepackage{comment}

\title{Randomization Helps in Online Graph Exploration: \\ Breaking the Deterministic Lower Bound on Cycles}
\author{
  Júlia Baligács\thanks{University of Oxford. Email: \texttt{jbaligacs@gmail.com}. Funded by the European Union through the European Research Council under the project BOBR (grant agreement No.~948057) during employment in Warsaw and under the project CCOO (grant agreement No.~101165139) during employment in Oxford. Views and opinions expressed are however those of the authors only and do not necessarily reflect those of the European Union or the European Research Council. Neither the European Union nor the granting authority can be held responsible for them.} 
  \and Jan Hązła\thanks{AIMS Rwanda. Email: \texttt{jan.hazla@aims.ac.rw}.
  Supported by the Alexander von Humboldt Foundation German research chair
funding and the associated DAAD project No. 57761435.} 
  \and Lena Volk\thanks{Technische Universität Darmstadt. Email: \texttt{volk@mathematik.tu-darmstadt.de}.}
}
\date{}

\newcommand{\alg}{\ensuremath{\textsc{Alg}}\xspace}
\newcommand{\opt}{\ensuremath{\textsc{Opt}}\xspace}

\newcommand{\dist}{\ensuremath{\textsc{Dist}}\xspace}
\newcommand{\simpledet}{\ensuremath{\textsc{HeavyTest}}\xspace}
\newcommand{\distrand}{\ensuremath{\textsc{RandHeavyTest}}\xspace}
\newcommand{\pdir}{p_{\mathrm{dir}}}
\newcommand{\Tstop}{T_{\mathrm{stop}}}
\newcommand{\N}{\mathbb{N}}
\newcommand{\R}{\mathbb{R}}
\DeclareMathOperator*{\EE}{\mathbb{E}}
\renewcommand{\epsilon}{\varepsilon}
\newcommand{\phistop}{\Phi_{\Tstop}}

\let\oldtextbf\textbf\renewcommand{\textbf}[1]{\oldtextbf{\boldmath #1}} 

\theoremstyle{definition}
\newtheorem{definition}{Definition}

\theoremstyle{plain}
\newtheorem{theorem}[definition]{Theorem}
\newtheorem{corollary}[definition]{Corollary}
\newtheorem{lemma}[definition]{Lemma}

\newtheorem{observation}[definition]{Observation}

\newtheorem{claim}[definition]{Claim}

\usepackage{hyperref}
\usepackage[nameinlink]{cleveref}

\crefformat{lemma}{#2Lemma~#1#3}
\crefformat{theorem}{#2Theorem~#1#3}
\crefformat{claim}{#2Claim~#1#3}
\crefformat{proposition}{#2Proposition~#1#3}
\crefformat{equation}{#2(#1)#3}
\crefformat{corollary}{#2Corollary~#1#3}
\crefformat{definition}{#2Definition~#1#3}
\crefformat{figure}{#2Figure~#1#3}
\crefformat{observation}{#2Observation~#1#3}
\crefformat{section}{#2Section~#1#3}
\crefformat{question}{#2Question~#1#3}

\definecolor{LinkViolet}{hsb}{0.37,1,0.5}
\hypersetup{
  colorlinks=true,
  linkcolor=Violet,
  citecolor=Violet,
  urlcolor=green!50!black
}
\tikzset{
  vertex/.style={circle, draw, fill=white, minimum size=4mm, inner sep=0pt,
                 inner sep=0pt, font=\small},
  cur/.style={circle, draw, fill=black!18, minimum size=4.5mm,
                 inner sep=0pt, font=\small},
  epath/.style={thick, decorate,
                 decoration={snake, amplitude=0.5mm, segment length=3mm,
                             pre length=2mm, post length=2mm}},
  cost/.style={fill=white, inner sep=1.5pt, font=\small}
}

\begin{document}

\maketitle
\thispagestyle{empty}

\begin{abstract}
In online graph exploration, introduced by Kalyanasundaram and Pruhs (1994), an agent must visit all vertices of an initially unknown weighted graph and return to its starting position, while the graph is revealed only locally at visited vertices. Although the problem has attracted considerable attention, previous work has focused exclusively on deterministic algorithms. Randomized strategies are often substantially harder to analyze because of a fundamental challenge inherent to exploration. In this work, we give the first positive result showing that randomization can improve competitive guarantees in online graph exploration. To this end, we focus on cycles, a simple graph class which nevertheless captures a key difficulty of online exploration.

Our main contribution is $\textsc{RandHeavyTest}$, a randomized algorithm for online exploration of cycles whose competitive ratio we prove to be at most \(1.315\). This establishes a strict separation from the deterministic setting, where the optimal competitive ratio is \(\thickapprox 1.366\), and thus gives the first provable advantage of randomization in online graph exploration. A key step towards this result is a new, simplified optimal deterministic algorithm, \(\textsc{HeavyTest}\), whose formulation naturally suggests the randomized variant. We complement our upper bounds with lower bounds of \(1.115\) for arbitrary randomized algorithms and \(1.207\) for the natural class of so-called forward-greedy algorithms, which includes \(\textsc{RandHeavyTest}\).
\end{abstract}



\section{Introduction}

In the classical online graph exploration problem, introduced by Kalyanasundaram and Pruhs in~1994~\cite{Kalyanasundaram94}, a single agent must explore an undirected, connected graph~$G=(V,E)$ with non-negative edge weights $w \colon E \to \R_{\geq 0}$.
The agent starts at some vertex $s\in V$ and initially has no knowledge of the graph. Whenever it visits a vertex for the first time, it learns the identifiers of all adjacent vertices together with the weights of the corresponding incident edges. The agent may then decide which known edge to traverse next. The cost of traversing an edge is simply its weight. The goal is to visit all vertices and return to $s$, subject to minimizing the total cost.

As usual, the quality of an exploration algorithm is measured in terms of competitive analysis.
For a deterministic algorithm $\alg$, let $\alg(G,s)$ denote the total cost incurred by $\alg$ on graph $G$ when started at vertex $s$.
Let $\opt(G)$ denote the offline optimum cost, that is, the minimum length of a closed walk in $G$ that visits every vertex.
Note that $\opt(G)$ is independent of the starting vertex, whereas the cost of an online algorithm may depend on it.
A deterministic algorithm $\alg$ is \emph{{$\rho$-competitive}} if  $\alg(G,s) \leq \rho \cdot \opt(G)$ for every connected graph $G=(V,E)$ and every starting vertex~$s\in V$.
The infimum over all such values $\rho$ is the \emph{competitive ratio} of~\alg.\footnote{
This is the \emph{strict} notion of competitiveness. The variant allowing an additive constant, $\alg(G,s) \leq \rho\cdot\opt(G)+C$, is known to yield the same infimum in online graph exploration~\cite{Baligacs26}. We therefore omit the qualifier ``strict''.
}

Arguably the most compelling question in online graph exploration is whether there exists a constant-competitive algorithm.
Equivalently, does the metric traveling salesperson problem admit a constant-factor approximation when the algorithm has no prior information about the graph?
Despite the naturalness of the question and considerable efforts to solve it~(see~\cref{sec:related_work} for a more detailed discussion, or~\cite[Chapter 2]{BaligacsThesis} for a survey), the problem has remained wide open since it was posed in~1994~\cite{Kalyanasundaram94}. The best known general upper bound on the competitive ratio is $O(\log n)$~\cite{Rosenkrantz77}, where $n$ denotes the number of vertices, while the best known lower bound is~$4$~\cite{Baligacs26}. The question has so far mainly been approached by identifying graph classes on which a constant competitive ratio is possible~\cite{Kalyanasundaram94,Megow2012,BaligacsDHS26} and by finding the exact competitive ratio in such classes~\cite{miyazaki,Unicyclic,TadpoleGraphs,Fritsch21}.

In this work, we initiate the study of randomized algorithms for online graph exploration.
A randomized algorithm $\alg$ is \emph{$\rho$-competitive} if~$\EE[\alg(G,s)] \leq \rho \cdot \opt(G)$ for every connected graph~$G$ and every starting vertex $s$.
In this formula, $G,s$ and $\opt(G)$ are fixed, whereas~$\alg(G,s)$ is a random variable.
Since randomization is a central paradigm in online algorithms, it is natural to ask whether it can improve the competitive guarantees in this setting.
This question was posed explicitly for cycles in~\cite{miyazaki}.
Yet, despite extensive work on exploration, no such improvement has been proved so far.
The same is true for several closely related and well-studied variants, including exploration with a team of agents~\cite{Cosson24, CossonMas24, disserlower, dereniowski, AkkerBF24} and exploration of directed graphs.
For the latter, randomized algorithms have been mentioned as a desirable direction~\cite{FleischerT05}; however, existing results concern lower bounds~\cite{DengPapadimitriou1999, FoersterW16} or experimental studies~\cite{FleischerT03experiments}, and no improved competitive guarantee is~known.

We believe that this absence of results is not accidental but reflects a difficulty inherent to exploration problems.
In many classical online problems (such as $k$-server, ski rental, secretary problem, online bipartite matching, online bin packing), the information available to the algorithm at a given time is independent of its earlier random choices (under the standard model of a so-called oblivious adversary).
In exploration this breaks down.
The part of the graph that has been revealed depends on the agent's previous trajectory, and hence on its previous random choices.
Thus, randomization affects not only the algorithm's decisions, but also the information
on which they are based.
This makes even seemingly simple randomized strategies notoriously difficult to analyze. Moreover, it is not a priori clear whether randomization can improve on the optimal deterministic competitive ratio for any natural class of graphs at all.

Our contribution is to develop techniques that handle this difficulty in the case of cycles.
Using these techniques, we prove that randomization beats the best possible deterministic competitive guarantee for this class. This gives the first separation between deterministic and randomized competitive guarantees in the Kalyanasundaram--Pruhs model. Thus, randomization is beneficial already on one of the simplest natural graph classes.

Observe that trees are trivial to explore with depth-first-search being 1-competitive. In case of cycles, if a cycle contains an edge that is more expensive than all other edges combined, an optimal offline tour avoids it and traverses the rest of the cycle twice. If there is no such edge, an optimal offline tour traverses the whole cycle once. However, in the online setting, whether an edge has this ``heavy'' property cannot be inferred when it is first discovered.  Therefore, cycles constitute a simple case which nevertheless captures a key difficulty of exploration and results in  a nontrivial problem.

Even in the deterministic case, settling the competitive ratio for online exploration of cycles
was far from immediate. The first progress was made in~\cite{AsahiroMMY10}, where the authors adapted the classical Nearest Neighbor algorithm~\cite{Rosenkrantz77} specifically to cycles, proving an upper bound of~$3/2$ and a lower bound of $5/4$ on the competitive ratio.
Miyazaki, Morimoto, and Okabe~\cite{miyazaki} later established a tight bound of $(1+\sqrt{3})/2 \approx 1.366$, thus closing the problem in the deterministic setting.
As mentioned before, the authors raised the question of whether randomization can improve on this bound. Although this was almost two decades ago, the question has remained open.
In this work, we provide a positive answer.

\subsection{Our results}

Our first contribution is a new deterministic algorithm for cycle exploration, \simpledet, that attains the optimal deterministic competitive ratio. Its main appeal lies in its simplicity. Unlike the previously known optimal algorithm $\textsc{Dist}$~\cite{miyazaki}, \simpledet does not base its decisions on the cost incurred so far, and its analysis is substantially shorter.

\begin{restatable}{theorem}{thmsimpledet}
\label{thm:simpledet}
The deterministic algorithm \simpledet has competitive ratio $\frac{\sqrt{3}+1}{2}\thickapprox 1.366$.
\end{restatable}

The formulation of \simpledet naturally suggests a randomized variant. From this starting point, we define the randomized algorithm $\distrand_\alpha$ parametrized by a value $\alpha>0$.
Our main result is the following upper bound on its competitive ratio for $\alpha=\nicefrac{1}{2}$.

\begin{restatable}{theorem}{thmdistrand}
\label{thm:distrand}
The randomized algorithm \(\distrand_{0.5}\) has competitive ratio at most\linebreak
\({1+\frac{\left(3-\sqrt{2}\right)^2}{8} \thickapprox 1.314}.\)
\end{restatable}

This establishes a strict separation between deterministic and randomized algorithms for online cycle exploration, since \(\distrand_{0.5}\) achieves a competitive ratio strictly below the deterministic optimum of 1.366.
This is the first result showing that randomization helps for online graph exploration (already on the class of cycles).

Throughout this introduction, we use the standard definition of the competitive ratio for randomized online algorithms against an \emph{oblivious adversary}.
In this model, the cycle is fixed independently of the random choices made by the algorithm. Our upper bound remains valid also against a stronger adversary that is allowed to adapt in a limited way to the outcomes of these random decisions.

Next, we complement this upper bound with a lower bound for arbitrary randomized algorithms.

\begin{restatable}{theorem}{thmlowerbound}
\label{thm:lowerbound}
Every randomized algorithm for online graph exploration on cycles has competitive ratio at least~\(1+\frac{\sqrt{2}}{8+3\sqrt{2}}\thickapprox 1.116.\)
\end{restatable}

Last, we provide a stronger lower bound for a natural class of algorithms called \emph{forward-greedy} that our algorithm \(\distrand_\alpha\) falls into.

\begin{restatable}{theorem}{thmlowerboundspecific}
\label{thm:lowerbound-specific}
Every forward-greedy randomized algorithm for online graph exploration on cycles has competitive ratio at least $\frac{1+\sqrt{2}}{2}\thickapprox 1.207$.
\end{restatable}

\subsection{Related work}
\label{sec:related_work}

A basic strategy for online graph exploration is the greedy \emph{Nearest Neighbor} algorithm, whose competitive ratio is $\Theta(\log n)$~\cite{Rosenkrantz77}, where the lower bound already holds on trees~\cite{Fritsch21} and on unweighted ladder graphs~\cite{HougardyW15}.
The \emph{Hierarchical Depth-First Search} algorithm is constant-competitive on graphs with a bounded number of distinct edge weights and has competitive ratio~$\Theta(\log n)$ on general graphs~\cite{Megow2012}. The algorithm $\textsc{Blocking}$ is constant-competitive on planar graphs~\cite{Kalyanasundaram94, Megow2012}; this was generalized first to bounded-genus graphs~\cite{Megow2012} and later to all graph classes excluding a fixed minor~\cite{BaligacsDHS26}. For general graphs, however, $\textsc{Blocking}$ does not improve upon the best known $\mathcal{O}(\log n)$ upper bound~\cite{Megow2012}, and there is currently no candidate algorithm conjectured to break the logarithmic barrier.
Lower bounds for arbitrary algorithms started at 2~\cite{miyazaki} and were successively improved~\cite{DobrevKM12,BirxDHK21} up to 4~\cite{Baligacs26}.

For more restricted graph classes, apart from cycles, the deterministic competitive ratio is known to be 2 on unweighted graphs~\cite{miyazaki} and on tadpole graphs~\cite{TadpoleGraphs}.
For unicyclic and cactus graphs, the best known upper bounds are 2.5~\cite{Unicyclic} and 3.91~\cite{Fritsch21}, respectively.
Another line of work studies online graph exploration with predictions~\cite{EberleLMNS22, Gehnen26}.

There are two main variants of directed graph exploration: Visiting all vertices in a weighted digraph, and traversing all edges in an unweighted digraph.
In the vertex-exploration model, the deterministic competitive ratio is exactly $n-1$ and the randomized ratio is lower bounded by $n/4$~\cite{FoersterW16}.
In the edge-exploration model, bounds are usually expressed in terms of the Eulerian deficiency $d$, the minimum number of edges that must be added to make the graph Eulerian.
The best known lower bounds are $\Omega(d)$ for deterministic and $\Omega(d/\log d)$ for randomized algorithms~\cite{DengPapadimitriou1999}, while the best known upper bound is $\mathcal{O}(d^8)$~\cite{FleischerT05}.
In both settings, all upper bounds are attained by deterministic algorithms.

In collaborative graph exploration, the objective is typically to minimize exploration time for unit-speed agents rather than total traveled distance.
Even unweighted trees remain challenging and are not fully understood for all team sizes \cite{Cosson24, CossonMas24, dereniowski, disserlower}. Cycles have also been studied and already exhibit nontrivial behavior~\cite{AkkerBF24}.

Finally, randomization is known to help in related online navigation problems, but these typically assume that the underlying metric space or graph is known and only obstacles are revealed online. Although such models also couple the available information with the agent's decisions, their techniques do not seem to transfer directly to the exploration settings discussed above. Examples include reaching a target in Euclidean space with unknown obstacles revealed only upon encounter~\cite{BermanBFKRS96, BlumRaghavanSchieber1997}, and the Canadian traveler problem ($k$-CTP), where an agent must reach a target in a known graph with at most $k$ unknown blocked edges~\cite{DemaineHLS21, BenderW15}.

\subsection{Outline and overview of techniques}

In \cref{sec:algorithms}, we first formally define the deterministic algorithm \simpledet, and then its natural randomization $\distrand_{\alpha}$.
The proof that \simpledet attains the optimal deterministic competitive ratio (\cref{thm:simpledet}) is deferred to \cref{sec:simpledet}.

Our main contribution is developed in \cref{sec:distrand}, where we prove the improved competitive ratio for randomized exploration of cycles (\cref{thm:distrand}).
The proof is based on a potential function argument, that is, we define an auxiliary process that combines the cost incurred so far by $\distrand_{0.5}$ with a carefully chosen potential of the current exploration state. The main technical step is to prove that this process has non-positive expected drift conditioned on the current state, and hence is a supermartingale.
We stop the process at a bounded stopping time~$\Tstop$, after which the algorithm makes at most one further randomized decision. The optional stopping theorem then gives a bound on the expected potential at time $\Tstop$. Finally, we bound the remaining expected cost of the algorithm, conditioned on the state at $\Tstop$, in terms of this potential. Together, these estimates yield the desired upper bound on the expected cost, and hence the claimed competitive ratio.

In \cref{sec:strategy}, we present this strategy in a formal way and identify the conditions on the potential function that suffice to prove \cref{thm:distrand}. In that section, we also discuss two notable features of the analysis. First, the argument remains valid against a stronger adversary that may adapt the cycle in a limited way to the outcomes of random decisions. Second, it is somewhat surprising that a potential depending only on information available to the algorithm can yield a better-than-deterministic bound. The stopping time $\Tstop$ is the key ingredient that makes this possible. In \cref{sec:potentialfunction}, we define the potential function and prove that it satisfies the conditions established in \cref{sec:strategy}. This is the technically most involved part of the paper.

Finally, in \cref{sec:lowerbounds}, we prove our lower bounds, Theorems~\ref{thm:lowerbound} and~\ref{thm:lowerbound-specific}. We proceed by an application of Yao's principle~\cite[Section 8.3]{yaosprinciple} to a small family of carefully chosen cycles.

\paragraph{Comment on AI usage.}
During the proof-discovery process of \cref{thm:distrand}, the authors used the large language model GPT-5.5 Pro. The tool was used in an iterative exchange to refine approaches for analyzing the algorithm $\distrand_{\alpha}$. In particular, it contributed to the identification of the potential function presented in \cref{sec:potentialfunction}, especially the function $g$ in it.


\section{The algorithms \simpledet and \distrand}
\label{sec:algorithms}

Before defining the algorithms, we fix the basic terminology used to describe the process of online graph exploration of cycles.
A \emph{boundary edge} is an edge that has one explored and one unexplored endpoint, in particular, its weight is already known to the agent, but the edge itself has not yet been traversed. The process of exploring a graph online can be thought of as follows: 
in each step, the agent chooses a boundary edge, travels through the explored part of the graph to its explored endpoint, traverses the edge, and then learns the information revealed at the newly visited vertex. For cycles, the explored part of the graph is always a path until the entire graph is known. Consequently, there are exactly two boundary edges, one at each end of this path. Moreover, after each exploration step, the agent is located at the explored endpoint of one of these two boundary edges. We call this the \emph{direct boundary edge}, and we call the other the \emph{backtracking boundary edge}. If the agent chooses to traverse the latter, we say that it \emph{backtracks}, and otherwise, we say that it moves~\emph{directly}.

In our algorithms, we assume that the first edge chosen is always the lighter of the two edges visible to the agent from the starting position; a tie is broken arbitrarily (deterministically). Once the last vertex has been visited, the agent has complete information and returns to the starting vertex via a shortest path.
In fact, observe that the agent obtains complete information one step earlier, when it discovers the last edge. At this time, there are two boundary edges known to have the same unexplored endpoint. However, we do not use this observation in our algorithms and simply assume that the agent is not aware at this time that the endpoints coincide.

At all other times, i.e., from the second step until all vertices are explored,
our algorithms make a decision based on the following three values (cf.~\cref{fig:explore-mid}). If the agent is currently located at $v$, let~$b$ denote the weight of the direct boundary edge, let $a$ denote the distance from $v$ to the starting position $s$, and let $d$ denote the distance from $s$ to the unexplored endpoint of the backtracking boundary edge.
With this notation, the cost of a direct move is $b$, whereas the cost of backtracking is $a+d$.
Upon visiting the last vertex, we define $a$ and $d$ to be the lengths of the two paths to $s$ (cf.~\cref{fig:explore-done}). It can be useful to think of this as a state in mid-exploration, where $b=0$ and the two boundary edges are connected by an edge of weight 0.

We will index the values $a,b,d$ by time, i.e., we write $a_t, b_t, d_t$, where we assume that at time~$t$ the algorithm has explored $t+1$ vertices. Note that $a_t,b_t,d_t$ are well defined only starting from~$t=1$.

The deterministic algorithm \simpledet is now defined by a simple threshold rule.

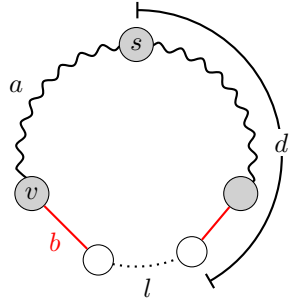
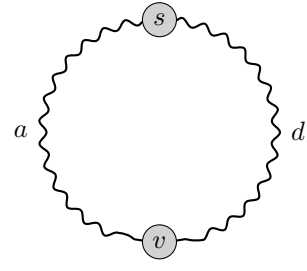
\begin{figure}
  \centering
  \begin{subfigure}[b]{0.52\textwidth}
    \centering
    \begin{tikzpicture}[scale=0.7]
      \def\R{2.1}
      \draw[epath] (90:\R) arc (90:205:\R);    
      \draw[epath] (90:\R) arc (90:-28:\R);    
      \draw[|-|,thick] (90:2.75) arc (90:-60:2.75);
      \draw[dotted, thick] (250:\R) arc (250:300:\R);
      \node[cur] (s)   at (90:\R)  {$s$};
      \node[cur] (cur) at (200:\R) {$v$};
      \node[cur] (back) at (-20:\R) {};
      \node[vertex] (P)   at (250:\R) {};
      \node[vertex] (Q)   at (300:\R) {};
      \draw[thick, red] (cur) -- node[cost, auto, swap] {$b$} (P);
      \draw[thick, red] (back) to (Q);
      \node[cost] at (150:2.62) {$a$};
      \node[cost] at (5:2.75)  {$d$};
      \node[cost] at (275:2.5) {$l$};
      \node[font=\scriptsize] at (208:2.95) {};
    \end{tikzpicture}
    \caption{Mid-exploration: paths of lengths $a$ and~$d$, two boundary edges depicted in red including the direct edge of length $b$,
             and the invisible remainder of length~$l$.}
    \label{fig:explore-mid}
  \end{subfigure}
  \hfill
  \begin{subfigure}[b]{0.4\textwidth}
    \centering
    \begin{tikzpicture}[scale=0.7]
      \def\R{2.1}
      \draw[epath] (90:\R) arc (90:270:\R);    
      \draw[epath] (90:\R) arc (90:-90:\R);    
      \node[cur] (s) at (90:\R)  {$s$};
      \node[cur]    (w) at (270:\R) {$v$};
      \node[cost] at (180:2.62) {$a$};
      \node[cost] at (0:2.62)   {$d$};
      \node[font=\scriptsize] at (270:2.95) {};
    \end{tikzpicture}
    \caption{After the last vertex is visited: only the two known
             paths $a$ and $d$ remain.\\ \phantom{asdf}}
    \label{fig:explore-done}
  \end{subfigure}
  \caption{Snapshots during graph exploration. The agent starts at vertex \(s\) and its current position is $v$.
  The gray vertices and the curled paths are already explored.
  }
  \label{fig:explore}
\end{figure}

\begin{definition}[\simpledet]
\label{def:simpledet}
The algorithm \simpledet is defined by letting the agent move directly if and only if $b\leq \sqrt{3}\,a+d$.
\end{definition}

In contrast to the previously known optimal algorithm \textsc{Dist}~\cite{miyazaki}, whose decisions depend on the cost incurred so far, \simpledet is history-independent in the following sense: its next move is determined solely by the values of $a,b,d$, regardless of how the current state was reached. In particular, the decision depends only on the total lengths $a$ and $d$ of the explored paths, and not on their decomposition into edges. Although it may not be surprising that an optimal rule with this independence property exists, it is far from clear that it can have such a simple description.

Note that the condition $b\leq \sqrt{3}\,a+d$ is equivalent to requiring that either $b\le a+d$ or $\frac{a}{b-a-d}\cdot (\sqrt{3}-1)\ge 1$. This reformulation suggests a natural way to randomize the deterministic rule.

\begin{definition}[\distrand]
\label{def:distrand}
The algorithm $\distrand_\alpha$ is parametrized by $\alpha>0$ and is defined by letting the agent move directly with probability
\begin{align*}
\pdir(a,b,d)\coloneqq\begin{cases}
1&\text{if }b\le (\alpha+1)a+d,\\
\frac{a}{b-d-a}\cdot\alpha&\text{if }b>(\alpha+1)a+d.
\end{cases}
\end{align*}
\end{definition}

Observe that \simpledet can be viewed as the deterministic rule obtained from the randomized algorithm $\distrand_{\sqrt{3}-1}$ by backtracking whenever the randomized algorithm does so with positive probability.

Next, note that, for every $\alpha>0$, the algorithm $\distrand_{\alpha}$ moves directly with probability~1 whenever it reaches a state with $b\leq a+d$.
We call any algorithm with this property \emph{forward-greedy}.
Recall that the offline optimum avoids traversing an edge only if its length exceeds the sum of the lengths of all other edges.
Thus, in a state with $b\leq a+d$, the online algorithm knows that the $b$ edge
is traversed in an optimum offline solution.
Nevertheless, it is unclear whether an optimal randomized algorithm for cycle exploration should necessarily be forward-greedy.
Viewed as a game against the adversary, it may be advantageous for the algorithm to be less predictable.


\section{Analysis of \distrand}
\label{sec:distrand}

In this section, we prove our main result, which strictly separates randomized cycle exploration from the deterministic case. We begin by recalling the statement.

\thmdistrand*

\subsection{Proof strategy}
\label{sec:strategy}

First, let us settle some basic terminology.
Given a cycle $\Gamma$, let $e^*$ denote an edge of maximum weight, i.e., $e^*\in \mathrm{argmax}\{w(e):e\in \Gamma\}$. Observe that the cost of the offline optimum is then
\[\opt(\Gamma)=\min\left\{\sum_{e\in \Gamma} w(e), \; 2 \cdot \sum_{e\in \Gamma\setminus \{e^*\}} w(e) \right\}.\]
In case $w(e^*)> \sum_{e\in \Gamma\setminus \{e^*\}} w(e)$, i.e., $\opt(\Gamma)$ is the latter argument of the minimum, we say that~$e^*$ is a \emph{heavy edge}. Otherwise, we say that the cycle does not have a heavy edge.

Given a randomized algorithm and an input cycle $\Gamma$, we define the random variable $\Tstop$ to be the time step defined as follows.
\begin{itemize}
\item If $\Gamma$ does not have a heavy edge, $\Tstop$ is the time when the last vertex is visited.
\item If $\Gamma$ contains a heavy edge, $\Tstop$ is the time when the heavy edge is discovered.
\end{itemize}

Note that $\Tstop$ is a stopping time with respect to the filtration $(\mathcal{F}_t)_{t\ge 0}$ generated by the history of the algorithm up to time $t$. In other words, the event $\Tstop=t$ only depends on the process up to time $t$.

We next observe that time $\Tstop$ is the last step in which the behavior of $\distrand_{\alpha}$ may be non-deterministic.

\begin{observation}
\label{obs:tstop}
From time $\Tstop+1$ on, the behavior of a forward-greedy algorithm, in particular $\distrand_{\alpha}$, is deterministic.
\end{observation}

\begin{proof}
In case there is no heavy edge, after time $\Tstop$, the algorithm simply returns via a shortest path to the starting position.
In case there is a heavy edge $e^*$, after the agent made a decision at time $\Tstop$, we never again encounter a state with $b>a+d$ because $b$ is the weight of an edge different from $e^*$ and we have $a+d\geq w(e^*)> b$.
Hence, the agent moves directly deterministically in every decision after~$\Tstop$ and returns to the starting position via a shortest path once all vertices are visited.    
\end{proof}

Observe that, in case the heavy edge is incident to the starting position, we have $\Tstop=0$, which may cause difficulties in our analysis, as $a,b,d$ are undefined before the first traversal. However, it is immediate that $\distrand_{\alpha}$ has competitive ratio 1 on such cycles, so we can ignore this case.
We call a cycle (together with a starting position) \emph{nontrivial} if it has at least 3 vertices and, in case it has a heavy edge, it is not incident to the starting position. In particular, $a,b,d$ are defined at time $\Tstop$ on nontrivial cycles.

The next lemma should be read mainly as a roadmap for the proof of \cref{thm:distrand}.
Although its proof is an immediate consequence of the optional stopping theorem, we include the details for completeness.
In the lemma and throughout the remainder of the paper, we simply denote the algorithm's total cost on a cycle by $\alg(\Gamma)$, omitting the starting vertex from the notation since it will always be clear from context.

\begin{lemma}
\label{lem:potential_function}
Let $\alg$ be a randomized forward-greedy algorithm for cycle exploration and $\rho\geq 1$.
Fix a nontrivial cycle $\Gamma$ and let
$(a_t,b_t,d_t,C_t)_{t\geq 1}$ denote the random process describing the values of $a,b,d$ at time $t$, and $C_t$ being the total cost incurred until time $t$.
Let $(\mathcal F_t)_{t\geq 0}$ be the filtration generated by the history of the algorithm up to time $t$. Assume there exists a function $\Phi \colon \R^3_{\ge 0} \to \R$ such that the process $\Phi_t := \Phi(a_t,d_t,C_t)$ satisfies the following properties.
\begin{enumerate}[label=(\alph*)]
\item $\Phi_1\leq 0$ with probability 1,
\label{cond:init}
\item $\EE[\alg(\Gamma) \mid \mathcal{F}_{\Tstop}] \leq \rho\cdot \opt(\Gamma) + \Phi_{\Tstop} $,
\label{cond:last_step}
\item \label{cond:supermartingale}
the stopped process $\Phi_t^*:=\Phi_{\min(t,\Tstop)}$ satisfies for every $t\geq 1$ 
\[
\EE[\Phi_{t+1}^* \mid \mathcal{F}_t]\leq \Phi_{t}^*.
\]
\end{enumerate}
Then, $\EE[\alg(\Gamma)]\leq \rho \cdot \opt(\Gamma)$.
\end{lemma}

\begin{proof}
Let us argue that $\EE[\Phi_{\Tstop}]\le 0$.  Since $\Tstop\le n$ (note that $\Gamma$ is fixed with $n$ vertices), and since, due to \ref{cond:supermartingale}, $(\Phi^*_t)_t$ is a supermartingale with respect to $(\mathcal{F}_t)_t$, this follows immediately by the optional stopping theorem for supermartingales (using property \ref{cond:init}). Alternatively, the inequality can be proved directly:
\begin{align*}
    \EE[\Phi_{\Tstop}]
    &= \EE[\Phi_n^*]
    = \EE\left[
        \Phi_1^* + \sum_{t=1}^{n-1}(\Phi_{t+1}^*-\Phi_t^*)
    \right] 
    = \EE[\Phi_1^*]
        + \sum_{t=1}^{n-1}\EE[\Phi_{t+1}^*-\Phi_t^*] \\
    &= \EE[\Phi_1^*]
        + \sum_{t=1}^{n-1}
        \EE\left[
            \EE[\Phi_{t+1}^*\mid \mathcal F_t] -\Phi_t^*
        \right]
        \overset{\ref{cond:supermartingale}}{\leq} \EE[\Phi_1^*]
        = \EE[\Phi_1]
        \overset{\ref{cond:init}}{\leq} 0.
\end{align*}
With this and condition~\ref{cond:last_step}, we obtain
\begin{align*}
\EE[\alg(\Gamma)]&=\EE \left[ \EE[\alg(\Gamma) \mid \mathcal{F}_{\Tstop}] \right]
\overset{\ref{cond:last_step}}{\leq}
\EE[\rho \cdot \opt(\Gamma) + \Phi_{\Tstop} ]\\
&= \rho \cdot \opt (\Gamma) + \EE[\Phi_{\Tstop}]
\leq \rho \cdot \opt(\Gamma).\qedhere
\end{align*}
\end{proof}

\paragraph{Comment on adaptive adversaries for online graph exploration.}
It is often useful to interpret an online problem as a two-player game: one player is the online algorithm, and the other is the adversary that presents a difficult instance.
The literature considers several adversary models, the most important of which can be summarized as follows. In the \emph{oblivious model}, the adversary constructs the entire problem instance in advance, knowing the algorithm's rules. In the \emph{adaptive offline model}, the adversary may decide only the part of the input that is currently presented to the algorithm. In particular, these choices may depend on the outcomes of previous (random) decisions. The offline optimum cost is then defined as usual, in hindsight.

In their seminal work, Ben-David, Borodin, Karp, Tardos, and Wigderson~\cite{BenDavidBKTW90} established several results comparing these adversary models.
First, deterministic algorithms have the same competitive ratio in the oblivious and adaptive offline models. Moreover, this deterministic competitive ratio also coincides with the randomized competitive ratio in the adaptive offline model. Although the class of problems considered in~\cite{BenDavidBKTW90} does not explicitly include exploration problems, it is not difficult to see that these results also apply to online graph exploration.

Now consider the strategy described in \cref{lem:potential_function}.
Observe that if there is a potential function~$\Phi$ satisfying conditions~\ref{cond:init} and~\ref{cond:supermartingale} against an oblivious adversary, then these conditions are also satisfied against an adaptive adversary. Indeed, this is immediate for condition~\ref{cond:init}, and it follows for condition~\ref{cond:supermartingale} since both sides of the inequality in \ref{cond:supermartingale} can be computed by the agent
at time $t$. Thus, at first sight, these conditions do not seem strong enough to yield a competitive ratio strictly below the deterministic optimum.
It may therefore seem surprising that a potential function depending only on values known to the algorithm can be used to prove such a result.
The key ingredient is the definition of the stopping time together with condition~\ref{cond:last_step}, which is what makes the technique work.

A closer inspection of the proof of \cref{lem:last_step}, where we verify condition~\ref{cond:last_step} for our potential function, shows that we require the adversary only to be oblivious from the time~$\Tstop$ onward.
Concretely, before the agent moves at any time~$t$, the adversary must commit to whether $t=\Tstop$ (this decision is however not revealed to the online algorithm).
If so, the adversary must also commit to the remaining cycle, consistent with the condition $t=\Tstop$. 
In particular, such an adversary has to decide immediately upon revealing an edge whether this edge will be heavy in the final cycle, that is, whether it will be included in the optimum offline solution.
This is similar in spirit to the model of the so-called \emph{adaptive online adversary} introduced in~\cite{BenDavidBKTW90}.

\subsection{The potential function}
\label{sec:potentialfunction}

To prove \cref{thm:distrand}, it suffices to define a potential function fulfilling the conditions of \cref{lem:potential_function}.
To this end, we define 
\begin{equation*}
    \Phi(a,d,C):=C-a-g(a,d),
\end{equation*}
where
\[g(a,d) \coloneq \max\left( 2rd, \frac{(4r-1)a+d}{2} \right) =\begin{cases}
        2rd, & a\leq d,\\[1mm]
        \frac{(4r-1)a+d}{2}, & a>d,
        \end{cases}
        \]
and $r:= \frac{\left(3-\sqrt{2}\right)^2}{8} \approx 0.314.$
In the following, we show that $\Phi$ fulfills conditions \ref{cond:init}--\ref{cond:supermartingale} of \cref{lem:potential_function} with $\rho:=1+r$ for the algorithm $\distrand_{0.5}$ for every nontrivial cycle.
For the remainder of the section, we fix a nontrivial cycle $\Gamma$ and analyze the process defined by the algorithm $\distrand_{0.5}$.
Let $(a_t,b_t,d_t,C_t)_t$ be as defined in \cref{lem:potential_function}.
We begin with condition~\ref{cond:init}, which is a simple observation.

\begin{observation}
\label{obs:potential-start-non-positive}
We have $\Phi_1\leq 0$.
\end{observation}

\begin{proof}
After the first edge traversal, we have $C_1=a_1$. Moreover, we have $g(a,d)\geq 0$ for all~$a,d\geq 0$.
Therefore, $\Phi_1=C_1-a_1-g(a_1,d_1)\leq 0$.
\end{proof}

Next, we show that condition~\ref{cond:last_step} is satisfied.

\begin{lemma}
\label{lem:last_step}
    We have $\EE[\alg(\Gamma) \mid \mathcal{F}_{\Tstop}] \leq (1+r)\cdot \opt(\Gamma) + \Phi_{\Tstop}$.
\end{lemma}

\begin{proof} 
We distinguish two cases, depending on the existence of a heavy edge.

\textbf{Case 1: \(\Gamma\) has no heavy edge.} 
At the time $\Tstop$ (cf.~\cref{fig:explore-done}), 
we have \(\opt(\Gamma)=a+d\) and
\begin{align*}
\alg(\Gamma)&=C+\min(a,d)
= C-a+a+\min(a,d)\\
 & \overset{(*)}{\leq} C-a+(a+d)(1+r)-g(a,d) \\
& = (a+d)(1+r) + C-a - g(a,d) = (1+ r)\cdot\opt(\Gamma) + \Phi_{\Tstop},
\end{align*}
where it is only left to prove the inequality $(*)$. This follows from the following auxiliary estimate for $g$:
\begin{equation*}
    \min(a,d) \leq d+r(a+d)-g(a,d)
\end{equation*} 
for all \(a,d \geq 0\). To see this, we distinguish two cases. If \(a\leq d\), then \(\min(a,d)=a\) and \(g(a,d)=2rd\) and we have \(a \leq (1-r)d+ra=d+r(a+d)-2rd\). Next, consider the case \(a >d\). Then, \(\min(a,d)=d\) and \(g(a,d)=\frac{(4r-1)a+d}{2}\). In this case, the desired inequality \(d \leq d + r(a+d)-\frac{(4r-1)a+d}{2}\) 
is equivalent to
$2r(a+d)\ge (4r-1)a+d$ and 
follows from
\[2r(a+d)
=(4r-1)a+(1-2r)a+2rd
\overset{a>d}{\geq}
(4r-1)a+(1-2r)d+2rd
= (4r-1)a+d.\]

\textbf{Case 2: \(\Gamma\) has a heavy edge.} 
Consider the situation at time $\Tstop$ when the heavy
edge is discovered (see \cref{fig:explore-mid}),
and let $l$ be the total weight of the edges which have not yet
been discovered.
If the agent moves directly, it incurs cost $b$ and then deterministically always proceeds directly, incurring another cost of $l+d$.
If it backtracks, it incurs cost $a+d$ and then deterministi\-cally~$2l+d$.
Moreover, it follows from the definition of~$\pdir$ (cf.~\cref{def:distrand}) that, since $b>a+d$, we have $\pdir(a,b,d)=\min(1,\alpha a/(b-a-d))$ so that $\pdir(a,b,d)(b-a-d)\leq \alpha a$.
Hence, we have
\begin{align*}
\EE[\alg(\Gamma) \mid \mathcal{F}_{\Tstop}] &\leq C+\pdir(a,b,d)(b+l+d)+(1-\pdir(a,b,d))
    (a+2d+2l) \\
&= C+a+2d+2l+\pdir(a,b,d)(b-a-d-l) \\
&\leq C+a+2d+2l+\pdir(a,b,d)(b-a-d) \\
&\leq C+a+2d+2l+\alpha a\\
&=(C -a)+2(a+d+l)+\alpha a.
\end{align*}
Since \(\opt(\Gamma)=2(a+d+l)\), and \(\alpha=\nicefrac{1}{2}\), we obtain
\begin{align*}
\frac{\EE[\alg(\Gamma) \mid \mathcal{F}_{\Tstop}]}{\opt(\Gamma)} \leq 1 + \frac{C-a}{\opt(\Gamma)} +\frac{a}{2\opt(\Gamma)}
=1+\frac{\phistop}{\opt(\Gamma)}+\frac{a+2g(a,d)}{4(a+d+l)}.
\end{align*}
Thus, it is only left to show that \(\frac{a+2g(a,d)}{4(a+d+l)} \leq r.\) This is true if \(a \leq 4r(a+d+l)-2g(a,d)\).
As~$l\geq 0$, it suffices to show
\begin{equation}
    \label{eq:second-estimate-g}
    a \leq 4r(a+d)-2g(a,d)
\end{equation}
for all \(a,d \geq 0\). To see this, we distinguish again two cases. If \(a\leq d\), then \(g(a,d)=2rd\) and we clearly have~\(a \leq 4ra\) since $r\geq \nicefrac{1}{4}$. If \(a>d\), we have \(g(a,d)=\frac{(4r-1)a+d}{2}\) and 
\[a = 4r(a+d)-((4r-1)a+4rd) \leq 4r(a+d)-((4r-1)a+d),\]
where we have again used $r\geq \nicefrac{1}{4}$. So \eqref{eq:second-estimate-g} holds, which completes the proof of the lemma in the case of a heavy edge.
\end{proof}

Last, but not least, we prove that condition~\ref{cond:supermartingale} holds.

\begin{lemma}
    \label{lem:supermartingale}
    For $\Phi_t^*:=\Phi_{\min(t,\Tstop)}$, we have
    \(
    \EE[\Phi_{t+1}^* \mid \mathcal{F}_t]\leq \Phi_{t}^*
    \)
    for every $t\geq 1$.
\end{lemma}

\begin{proof}
If $t\geq \Tstop$, the asserted inequality holds trivially with equality, so let \(t\) be some time before~\(\Tstop\). Let \(a,b,d,\) and \(C\) be the corresponding values at time~\(t\), i.e.,~\(a=a_t\), \(b=b_t\) and so on, and let \(a', b',d'\) and \(C'\) be the corresponding values at time \(t+1.\) We need to show 
   \begin{equation*}
       \EE[C'-a'-g(a',d') \mid \mathcal{F}_t]\leq C-a-g(a,d).
   \end{equation*}
   If the algorithm moves directly, then
    \[C'=C+b, \qquad a'=a+b, \qquad d'=d.\]
    If the algorithm backtracks, then
    \[C'=C+a+d, \qquad a'=d, \qquad d'=a+b.\]
    First suppose that \(b \leq 1.5a +d.\) Then \(\pdir(a,b,d)=1\) and the algorithm moves directly. Then \[C' - a' - g(a',d') = C - a - g(a+b,d) \leq C -a - g(a,d),\]
    where the inequality holds as the function \(g\) is increasing in the first variable.
    
    Now suppose that \(b > 1.5a +d.\) If \(a=0,\) the algorithm backtracks with probability 1. 
    In this case, 
    \[C'-a'-g(a',d')=C+a-g(d,a+b)=C-g(d,b) \leq C-a - g(a,d),\] as the function \(g\) is increasing in both variables and we have $d\geq 0=a$ and  \(b \geq d.\) 
    
    Hence we can assume in the following that \(a>0\) and \(p:=\pdir(a,b,d)=a/(2(b-a-d))\in (0,1).\) Since \(b \geq 1.5a+d\), 
    it follows $a+b\ge d$ and hence
    \[g(a+b,d)=\frac{(4r-1)(a+b)+d}{2}, \qquad g(d,a+b)=2r(a+b).\]
    So the expected change in \(\Phi\) is
    \begin{align*}
        \Delta &\coloneq \EE[\Phi'-\Phi\mid \mathcal{F}_t] \\
        &= p\bigl(C-a-g(a+b,d)\bigr)+(1-p)\bigl(C+a-g(d,a+b)\bigr)-\bigl(C-a-g(a,d)\bigr) \\
        &= 2a(1-p)-p\,g(a+b,d)-(1-p)g(d,a+b) +g(a,d) \\
        &= 2a(1-p) - p\,\frac{(4r-1)(a+b)+d}{2} - (1-p)2r(a+b) +g(a,d) \\
        &= g(a,d) + 2a - 2r(a+b) + p\left(2r(a+b)-2a - \frac{(4r-1)(a+b)+d}{2}\right) \\
        &= g(a,d) + 2a - 2r(a+b) + p\, \frac{b-3a-d}{2} \\
        &= g(a,d) + 2a - 2r(a+b) + \frac{a(b-3a-d)}{4(b-a-d)}.
    \end{align*}
    
    Now it only remains to prove that \(\Delta \leq 0\) by upper bounding the resulting expression for all possible values of \(a,b,d\) (i.e., $a>0, d\geq 0$ and $b\geq 1.5a+d$).
    A reader can convince themselves of this bound using
    a computer tool of their choice.
    Below, we provide an analytical proof.

    We distinguish two cases.
    
    \textbf{Case 1: \(a \leq d\).} Then \(g(a,d)=2rd.\) We have
    \begin{align*}
        \Delta &= 2rd + 2a - 2r(a+b) + \frac{a(b-3a-d)}{4(b-a-d)}
        = -2r(b-a-d) + \left(\frac{9}{4}-4r\right)a-\frac{a^2}{2(b-a-d)}.
    \end{align*}
    Now, dividing this equation by \(a>0\) we obtain
    \[\frac{\Delta}{a}=\left(\frac{9}{4}-4r\right)-2r\frac{b-a-d}{a}-\frac{a}{2(b-a-d)}.\]
    Let \(x:=\frac{2(b-a-d)}{a}>0.\)
    For all \(x>0\), we have 
    \[f(x):=\left(\frac{9}{4}-4r\right)-xr-\frac{1}{x}<0.\]
    To see this, note that the function achieves its maximum for $x> 0$ at the positive zero of its derivative~$f'(x)=-r+1/x^2$, namely at $x=\sqrt{1/r}$. Hence, $f(x)\leq f(\sqrt{1/r})\thickapprox -0.129$.
    So in particular, we have $\Delta\leq a\cdot f(x)< 0$.
    
    \textbf{Case 2: \(a > d\).} Then \(g(a,d)=\frac{(4r-1)a+d}{2}\) and we have
    \begin{align*}
        \Delta &= \frac{(4r-1)a+d}{2} + 2a - 2r(a+b) + \frac{a(b-3a-d)}{4(b-a-d)}
        = \frac{3}{2}a+\frac{1}{2}d-2rb+ \frac{a(b-3a-d)}{4(b-a-d)}.
    \end{align*}

To bound the remaining expression, we set \(x\coloneq b-a-d>0\) and make use of the following auxiliary estimate which follows from the inequality of arithmetic and geometric means
        \begin{equation}
            \label{eq:AM-GM}
            2rx+\frac{a^2}{2x} \geq 2\sqrt{2rx\frac{a^2}{2x}} = 2a\sqrt{r}.
        \end{equation}
        Overall, we obtain
        \begin{align*}
            \Delta&=\frac{3}{2}a+\frac{1}{2}d-2rb+ \frac{a(b-3a-d)}{4(b-a-d)} = \frac{3}{2}a+\frac{1}{2}d-2rb+ \frac{a}{4}-\frac{a^2}{2x} \\
            &= \left(\frac{7}{4}-2r\right)a+\left(\frac{1}{2}-2r\right)d-2rx-\frac{a^2}{2x} 
            \overset{\cref{eq:AM-GM}}{\leq} \left(\frac{7}{4}-2r\right)a+\left(\frac{1}{2}-2r\right)d - 2a\sqrt{r} \\
            &= \left(\frac{7}{4}-2r-2\sqrt{r}\right)a+\left(\frac{1}{2}-2r\right)d
            =0\cdot a + \left(\frac{1}{2}-2r\right)d
            \leq 0,
        \end{align*}
So in either case, we have $\Delta \leq 0$. This completes the proof of the lemma.
\end{proof}

Finally, putting the results of \cref{obs:potential-start-non-positive}, \cref{lem:last_step}, and \cref{lem:supermartingale} together, we obtain from~\cref{lem:potential_function} for \(\rho=1+r\) that~\cref{thm:distrand} holds. 


\section{Lower bounds for randomized algorithms}
\label{sec:lowerbounds}

We begin by proving a lower bound for the competitive ratio of any randomized algorithm (\cref{thm:lowerbound}) and then give an improved lower bound for a natural subclass of randomized algorithms to which the algorithms \(\distrand_\alpha\) belong (\cref{thm:lowerbound-specific}).

Let us denote a cycle on \(n\) vertices together with a starting vertex \(s\) by \((w_1,\dots,w_n),\) where~\(w_i\) denotes the weight of the \(i\)-th edge starting from \(s\) in one of the  directions. 
In particular, at the beginning of exploration, the agent sees two edges
with weights $w_1$ and $w_n$.

\thmlowerbound*

\begin{proof}
   We prove this statement using Yao's principle~\cite[Section 8.3]{yaosprinciple}, i.e., we give a set of cycles together with a probability distribution over them and prove that any deterministic algorithm has expected competitive ratio on a random cycle from this set of at least~\(1+\frac{\sqrt{2}}{8+3\sqrt{2}}\). By a straightforward application of Yao's principle, this then implies the theorem.

    Let \(x\coloneq 2+2\sqrt{2}.\) Consider the set of the following three cycles each on four vertices, illustrated in \cref{fig:lower-bound}:    
    \[\Gamma_1=(1,0,x,1), \; \Gamma_2=(1,x,0,1), \text{ and } \Gamma_3=(1,x,x,1).\] 

    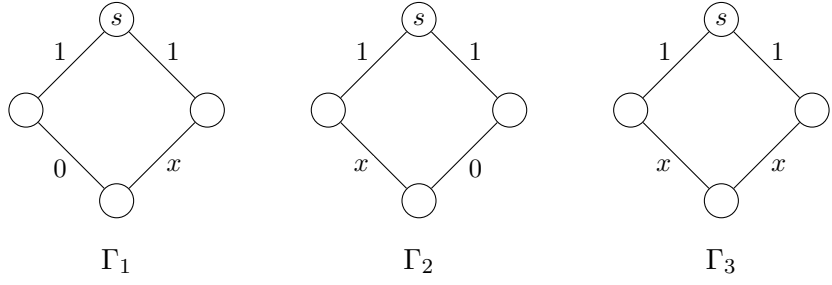
\begin{figure}
    \centering
    \begin{tikzpicture}[
        vertex/.style={circle, draw, fill=white, minimum size=4.5mm,
                       inner sep=0pt, font=\small},
        weight/.style={fill=white, inner sep=1.5pt, font=\small}
      ]
      \foreach \dx/\wA/\wB/\wC/\wD/\lbl in {%
          0/1/0/x/1/{\Gamma_1},
          4/1/x/0/1/{\Gamma_2},
          8/1/x/x/1/{\Gamma_3}}{
        \begin{scope}[xshift=\dx cm]
          \node[vertex] (s) at (0, 1.2)  {$s$};
          \node[vertex] (l) at (-1.2, 0) {};
          \node[vertex] (b) at (0, -1.2) {};
          \node[vertex] (r) at (1.2, 0)  {};
          \draw (s) -- node[weight, auto, swap] {$\wA$} (l);
          \draw (l) -- node[weight, auto, swap] {$\wB$} (b);
          \draw (b) -- node[weight, auto, swap] {$\wC$} (r);
          \draw (r) -- node[weight, auto, swap] {$\wD$} (s);
          \node at (0, -2) {$\lbl$};
        \end{scope}
      }
    \end{tikzpicture}
    \caption{The graphs from the proof of \cref{thm:lowerbound}.}
    \label{fig:lower-bound}
    \end{figure}
    
    Let \(p\coloneq\frac{2}{8+3\sqrt{2}}.\) Assume cycles \(\Gamma_1\) and \(\Gamma_2\) are each sampled with probability \(p\) and cycle \(\Gamma_3\) with probability~\(1-2p.\)
    
    First of all, note that \(\opt(\Gamma_1)=\opt(\Gamma_2)=4\) and \(\opt(\Gamma_3)=2+2x\) as \(x\geq 2\). 
Now let \(\alg\) be some deterministic algorithm for cycle exploration. 
Since all three cycles look exactly the same from the starting vertex, the first vertex visited by $\alg$ is the same in all three cycles and, by symmetry,
we can assume without loss of generality that the agent starts towards the \(x\)-edge in \(\Gamma_1\) (i.e., in any cycle in \cref{fig:lower-bound}, the agent visits the rightmost vertex first).
First, we simply note~$\alg(\Gamma_2)\ge \opt(\Gamma_2)$.
For the other two cycles, the agent encounters an edge of weight $x$ after its first traversal and then makes the same decision on both 
$\Gamma_1$ and $\Gamma_3$.
    
    If the agent traverses the first seen \(x\)-edge, then \(\alg(\Gamma_1) \geq 2+x\). If the agent does not traverse the first seen \(x\)-edge, then \(\alg(\Gamma_3)\ge 4+2x.\)
    Since cycles \(\Gamma_1, \Gamma_2\) are sampled with probability \(p\) and cycle~\(\Gamma_3\) is sampled with probability~\(1-2p\), we obtain for a random cycle \(\Gamma\) chosen from the set \(\{\Gamma_1,\Gamma_2,\Gamma_3\}\) according to the given probability distribution that
    \begin{align*}
        \EE\left[\frac{\alg(\Gamma)}{\opt(\Gamma)}\right] & \geq \min\left\{ p \cdot 1 + p \cdot 1 + (1-2p)\, \frac{4+2x}{2+2x}, \; p \, \frac{2+x}{4} + p \cdot 1 + (1-2p) \cdot 1\right\} \\
        &= 1 + \min\left\{ (1-2p)\,\frac{1}{1+x},\; p\,\frac{x-2}{4}\right\}
        = 1+\frac{\sqrt{2}}{8+3\sqrt{2}}. \qedhere
    \end{align*}
\end{proof}

Recall that a randomized algorithm for online cycle exploration is called \textit{forward-greedy} if the agent moves directly whenever \(b \leq a+d\).
(The algorithm can use any decision rule in the first step,
as well as after discovering the last vertex.)
For the class of forward-greedy algorithms, we obtain the following improved lower bound.

\thmlowerboundspecific*

\begin{proof}
    Let \(1>\epsilon>0\) such that $1/\epsilon^2\in \N$.
    Consider the following two cycles (cf.~\cref{fig:lower-bound-forward}): 
    \[
    \Gamma_1=(\epsilon^2,\epsilon^2,\dots,\epsilon^2,1+\sqrt{2},0,\epsilon), \quad  \Gamma_2=(\epsilon^2,\epsilon^2,\dots,\epsilon^2,1+\sqrt{2},2+\sqrt{2},\epsilon),
    \]
    where in each of them there are \(\nicefrac{1}{\epsilon^2}\) many edges of weight \(\epsilon^2.\)
    For \(i \in \{1,2\}\), let \(v_i\) be the vertex between the edge of weight \(\epsilon^2\) and the edge of weight \(1+\sqrt{2}\) in \(\Gamma_i\). Analogously, let \(u_i\) be the vertex between the edge of weight 0 (resp.~weight \(2+\sqrt{2}\)) and \(\epsilon\).
    
    \begin{figure}
    \centering
    \begin{tikzpicture}[scale=0.7,
        vertex/.style={circle, draw, fill=white, minimum size=5mm,
                       inner sep=0pt, font=\small},
        weight/.style={fill=white, inner sep=1.5pt, font=\small}
      ]
      \foreach \dx/\vlab/\ulab/\medge/\glab in {%
          0/{v_1}/{u_1}/{0}/{\Gamma_1},
          12/{v_2}/{u_2}/{2+\sqrt{2}}/{\Gamma_2}}{
        \begin{scope}[xshift=\dx cm]
          \node[vertex] (s) at (0, 3)     {$s$};
          \node[vertex] (u) at (2.2, 1.2) {$\ulab$};
          \node[vertex] (w) at (2.2,-1.2) {};
          \node[vertex] (v) at (0,-3)     {$\vlab$};
          \node[vertex] (c1) at (-1.8, 2.2) {};
          \node[vertex] (c2) at (-2.7, 0.9) {};
          \node[vertex] (c3) at (-2.1,-2.0) {};
          \draw (s)  -- node[weight, auto, swap] {$\epsilon^2$} (c1);
          \draw (c1) -- node[weight, auto, swap] {$\epsilon^2$} (c2);
          \draw[dotted, thick] (c2) to[bend right=20] (c3);
          \draw (c3) -- node[weight, auto, swap] {$\epsilon^2$} (v);
          \node[align=center, font=\small] at (-4.7, -0.55)
                {$1/\epsilon^2$ edges\\ of weight $\epsilon^2$};
          \draw (v) -- node[weight, auto, swap] {$1+\sqrt{2}$} (w);
          \draw (w) -- node[weight, auto, swap] {$\medge$} (u);
          \draw (u) -- node[weight, auto, swap] {$\epsilon$} (s);
          \node[font=\large] at (0,-4.0) {$\glab$};
        \end{scope}
      }
    \end{tikzpicture}
    \caption{The cycles from the proof of \cref{thm:lowerbound-specific}.}
    \label{fig:lower-bound-forward}
    \end{figure}
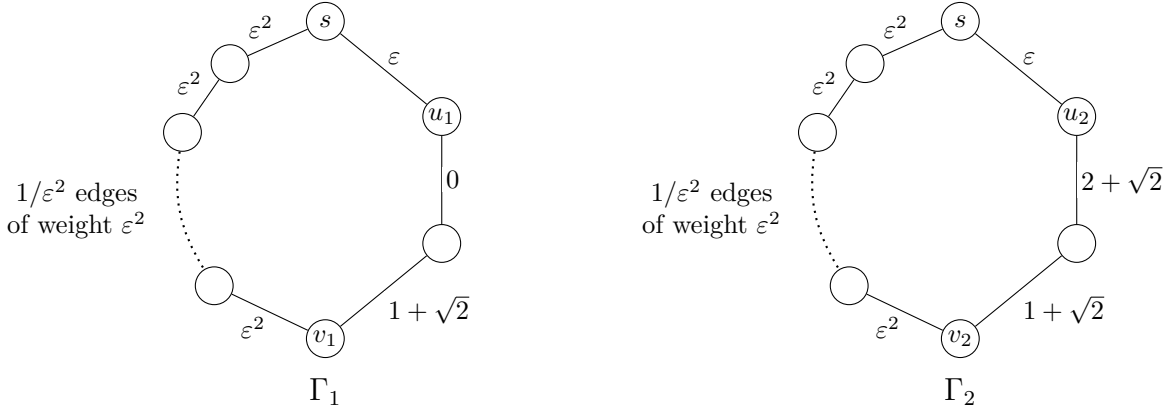
    Let \(\alg\) be a forward-greedy randomized algorithm for online cycle exploration. Hence, the agent moves directly whenever \(b \leq a+d\). As a first step, we prove the following. 
    \begin{claim}
        \label{cl:v-visited-before-u} If $\alg$ has competitive ratio at most 2, then for $i=1,2$,
        the vertex \(v_i\) is visited before~\(u_i\)
        with probability at least $1-4\epsilon$.
    \end{claim}
    \begin{proof}[Proof of~\cref{cl:v-visited-before-u}]
        Since the agent moves directly if \(b \leq a+d\), it only remains to show that \(\alg\) starts with the edge of weight \(\epsilon^2\) with probability at least $1-4\epsilon$.
        If $\alg$ begins by traversing the edge of weight~$\epsilon$ with probability $p$, and the edge of weight $\epsilon^2$ with probability $1-p$, this is also the case for the cycle $\Gamma=(\epsilon^2,0,\epsilon)$.
        Note that the edge of weight $\epsilon$ is a heavy edge for $\Gamma$ as $\epsilon^2<\epsilon$ (where we used $\epsilon<1$).
        We obtain
        \begin{equation*}
            2\geq \frac{\EE[\alg(\Gamma)]}{\opt(\Gamma)}\ge\frac{(1-p)\cdot 2 \epsilon^2+p\cdot (\epsilon + \epsilon^2)}{2 \epsilon^2}\geq \frac{p\epsilon}{2\epsilon^2}=\frac{p}{2\epsilon}.
        \end{equation*}
        Multiplying the inequality with $2\epsilon$ yields $p\leq 4 \epsilon$, which completes the proof of the claim.
    \end{proof}
    Assume for now that the agent is standing at vertex \(v_i\) and has not visited \(u_i\) yet. Let \(p_i\) be the probability with which the agent moves directly, i.e., traverses the \((1+\sqrt{2})\)-edge, in this setting. Note that \(p_1=p_2\) as we assume that \(u_i\) wasn't visited by the agent yet and hence, \(\alg\) has no knowledge about any difference in \(\Gamma_1\) and \(\Gamma_2.\) In the following, we denote this probability of moving directly by~\(p.\)
    
    First of all, note that \(\opt(\Gamma_1)=2+2\epsilon\) and \(\opt(\Gamma_2)=4+2\sqrt{2}+\epsilon.\) Further, we have
    \[\EE[\alg(\Gamma_1) \; | \; v_1 \text{ visited before } u_1]\geq p(2+\sqrt{2}+\epsilon)+(1-p)(2+2\epsilon)=2+2\epsilon+p\left(\sqrt{2}-\epsilon\right),\] 
    and 
    \begin{align*}
        \EE[\alg(\Gamma_2) \; | \; v_2 \text{ visited before } u_2] &\geq p(4+2\sqrt{2}+\epsilon)+(1-p)(6+2\sqrt{2}+\epsilon) \\
        & = 4+2\sqrt{2}+\epsilon+2(1-p).
    \end{align*}
    Let \(A_i\) be the event that the agent of \(\alg\) visits \(v_i\) before \(u_i\). By~\cref{cl:v-visited-before-u} we can assume that the probability for \(A_i\) is at least \(1-4\epsilon,\) as otherwise \(\alg\) is not even 2-competitive.
    Thus, for the competitive ratios we obtain
    \begin{align*}
        \EE\left[\frac{\alg(\Gamma_1)}{\opt(\Gamma_1)}\right] &=  \Pr\left[A_1\right] \cdot \EE\left[\frac{\alg(\Gamma_1)}{\opt(\Gamma_1)} \; \middle| \; A_1\right] + \Pr\left[\overline{A_1}\right] \cdot \EE\left[\frac{\alg(\Gamma_1)}{\opt(\Gamma_1)} \; \middle| \; \overline{A_1}\right] \\
        &\geq (1-4\epsilon)\cdot \left(1+p\,\frac{\sqrt{2}-\epsilon}{2+2\epsilon}\right),
    \end{align*}
    and analogously
    \[\EE\left[\frac{\alg(\Gamma_2)}{\opt(\Gamma_2)}\right]\geq (1-4\epsilon)\cdot\left(1 + (1-p)\frac{2}{4+2\sqrt{2}+\epsilon}\right).
    \]
    Hence, the competitive ratio of \(\alg\) is at least 
    \[\max\left\{1+(1-4\epsilon)\cdot \left(p\,\frac{\sqrt{2}-\epsilon}{2+2\epsilon}\right), \; 1+(1-4\epsilon)\cdot\left((1-p)\frac{2}{4+2\sqrt{2}+\epsilon}\right)\right\}.\]
    Since this holds for all \(\epsilon>0\) (with \(1/\epsilon^2\in \N\)), we obtain for \(\epsilon \to 0\) that the competitive ratio of \(\alg\) is at least~\(1+\max\left\{\frac{p}{\sqrt{2}}, \frac{1-p}{2+\sqrt{2}}\right\}\). As one term in this maximum is increasing in \(p\), while the other one is decreasing in \(p\), this maximum is minimized for the \(p\) such that \(\frac{p}{\sqrt{2}}=\frac{1-p}{2+\sqrt{2}},\) i.e.,~for~\(p=\frac{1}{2+\sqrt{2}}.\) Thus, the competitive ratio of \(\alg\) is at least 
    \[1+\max\left\{\frac{p}{\sqrt{2}}, \frac{1-p}{2+\sqrt{2}}\right\} 
    \geq 1+\max\left\{\frac{1}{\sqrt{2}\left(2+\sqrt{2}\right)}, \frac{1+\sqrt{2}}{\left(2+\sqrt{2}\right)^2}\right\}
    =\frac{1+\sqrt{2}}{2}. \qedhere\]
\end{proof}


\section{A simplified optimal deterministic algorithm}
\label{sec:simpledet}

In this section, we prove that our simpler deterministic algorithm \simpledet has competitive ratio~$\frac{\sqrt{3}+1}{2}$ (\cref{thm:simpledet}). By the result in~\cite{miyazaki}, this ratio is optimal for deterministic algorithms. Recall that in \simpledet the agent moves directly if and only if \(b\leq \sqrt{3}\,a+d\). In fact, we prove the following slightly more general statement.

\begin{theorem}
\label{thm:det-bound-comp-rat}
Let \alg be an algorithm (deterministic or randomized) for online graph exploration on cycles that satisfies the following properties for some \(\alpha>1\) and \(\beta \geq \alpha\):
\begin{enumerate}[label=(\roman*)]
    \item \label{it:suf-cond-direct} Whenever \(b \leq \alpha a +d\), then the agent moves directly.
    \item \label{it:suf-cond-backtr} Whenever \(b > \beta a + d\), then the agent backtracks.
\end{enumerate}
Also assume that $\alg$ initially traverses the lighter of the two visible edges, and that after visiting the last
vertex it returns to the
starting vertex in an optimal fashion.
Then the competitive ratio of~\alg is at most \(\max\left\{1+\frac{1}{1+\alpha},\frac{1+\beta}{2}\right\}\).
\end{theorem}

We first note in the following corollary that this result indeed implies \cref{thm:simpledet}.

\begin{corollary}
    The deterministic algorithm \simpledet has competitive ratio $\frac{\sqrt{3}+1}{2}$.
\end{corollary}

\begin{proof}
    The algorithm \simpledet satisfies the properties from \cref{thm:det-bound-comp-rat} for \(\alpha=\beta=\sqrt{3}.\) Thus, the competitive ratio of \simpledet is at most
    \[\max\left\{1+\frac{1}{1+\alpha},\frac{1+\beta}{2}\right\} = \max\left\{1+\frac{1}{1+\sqrt{3}},\frac{1+\sqrt{3}}{2}\right\}=\frac{\sqrt{3}+1}{2}.\qedhere\]
\end{proof}

In the following, let $C$ be the total cost incurred by the algorithm so far.
The following lemma was proved in~\cite{miyazaki} for their algorithm \dist and is also a key ingredient for our analysis.

\begin{lemma}
\label{lem:cstar-bound-gen}
Let \alg be an algorithm satisfying condition~\ref{it:suf-cond-direct} of \cref{thm:det-bound-comp-rat} for some \(\alpha>1\). Then~\(C -a\leq \frac{2}{1+\alpha}d\) at every point in time up to (and including) the discovery of the last vertex.
\end{lemma}

\begin{proof}
    We prove this by induction. After the first step, \(C=a\) and the inequality holds true. In case the agent moves directly, the value of \(C-a\) and the value of \(d\) remain unchanged. Thus, by induction hypothesis the inequality holds. In case the agent backtracks, we have \(b > \alpha a +d\) by condition~\ref{it:suf-cond-direct}. 
    Moreover, the change in the variables $a,C,d$ is $\Delta a= d-a, \Delta C=a+d, \Delta d=a+b-d$.
    Thus, we obtain for the change in \(C-a\) and the change in \(d\) that
    \[\frac{\Delta (C-a)}{\Delta d} = \frac{2a}{a+b-d} < \frac{2a}{a + \alpha a} = \frac{2}{1+\alpha}.\]
    By induction hypothesis, the asserted inequality holds after the move.
\end{proof}

Now we have all the prerequisites at hand to prove the main result of this section.

\begin{proof}[Proof of~\cref{thm:det-bound-comp-rat}]
Let \alg be an algorithm that satisfies 
the assumptions of the theorem. Let~\(\Gamma\) be a cycle. We split the analysis in two cases. First, assume that the cycle \(\Gamma\) has no heavy edge. Consider the time when the last vertex is visited, see \cref{fig:explore-done} for an illustration.
Let \(C, \, a,\) and~\(d\) be the respective values at this time. Then,
\begin{align*}
    \frac{\alg(\Gamma)}{\opt(\Gamma)}=\frac{C+\min\{a,d\}}{a+d}
    \overset{\text{Lem.}\ref{lem:cstar-bound-gen}}{\leq} \frac{\frac{2}{1+\alpha}d + a + \min\{a,d\}}{a+d}
    \overset{(*)}{\leq} \frac{\left(2+\frac{2}{1+\alpha}\right)a}{2a}
    = 1+ \frac{1}{1+\alpha},\nonumber
\end{align*}
where the inequality $(*)$ holds
since the expression 
on the left hand side
is, for fixed $a$, increasing in $d$
for $d\le a$ and decreasing in $d$ for $d\ge a$ (where we used $\alpha\geq 1$), i.e., it is maximized for $d=a$.

Now, we assume that the cycle has a heavy edge. If the heavy edge is visible at the beginning of the algorithm, then by the assumption that the algorithm initially traverses the lighter edge, the heavy edge will never be traversed and the competitive ratio is 1.

Otherwise, consider the moment \alg encounters the heavy edge for the first time, i.e., \(b\) is the heavy edge. 
Let \(l\) be the sum of the weights of all remaining edges not covered by \(a,b\), or \(d\), see \cref{fig:explore-mid}. Then, \(\opt = 2(a+d+l).\) Assume that the agent backtracks when encountering the heavy edge. Since after encountering the heavy edge, the algorithm never backtracks again by condition~\ref{it:suf-cond-direct}, the cost of the algorithm is \(C +a+2d+2l\). Together with \cref{lem:cstar-bound-gen} we obtain
\[\alg(\Gamma) = C +a+2d+2l \leq \left(\frac{2}{1+\alpha}\right)d + 2a + 2d+ 2l.\]
Hence, we have
\[\frac{\alg(\Gamma)}{\opt(\Gamma)} = \frac{\left(\frac{2}{1+\alpha}\right)d+2a+2d+2l}{2a+2d+2l} \leq 1+\frac{\left(\frac{2}{1+\alpha}\right)d}{2d} = 1+ \frac{1}{1+\alpha}.\]
It remains to consider the case in which the agent chooses to traverse the heavy edge. In this case, by condition~\ref{it:suf-cond-backtr}, we have \(b \leq \beta a + d.\)
After traversing the heavy edge, the agent never backtracks again by condition~\ref{it:suf-cond-direct}. Thus, the cost of the algorithm is \(C+b+l+d\) and we obtain with \cref{lem:cstar-bound-gen} that
\begin{align*}
    \alg(\Gamma) &= C+b+l+d 
    \leq \left(\frac{2}{1+\alpha}\right)d+a+\beta a+d+l+d
    = \left(\frac{2}{1+\alpha} +2 \right)d+(1+\beta )a+l.
\end{align*}
Hence, we have 
\begin{align*}
    \frac{\alg(\Gamma)}{\opt(\Gamma)} & =\frac{\left(\frac{2}{1+\alpha} +2 \right)d+(1+\beta )a+l}{2(a+d+l)}
    =1+ \frac{\frac{2}{1+\alpha}d+(\beta-1 )a-l}{2(a+d+l)}
    \leq 1+\frac{\frac{2}{1+\alpha}d+(\beta-1)a}{2(a+d)} \\
    & \le 1+\max\left\{\frac{1}{1+\alpha},\frac{\beta-1}{2}\right\}= \max\left\{1+\frac{1}{1+\alpha},\frac{1+\beta}{2}\right\}.\qedhere
\end{align*}
\end{proof}


\section{Concluding remarks}

In this paper, we gave the first provable advantage of randomization in the online graph exploration model introduced by Kalyanasundaram and Pruhs. 
We believe that extending this result beyond cycles and understanding whether randomization can help for more complex classes of graphs is a challenging but promising direction for further research.

For cycles, it remains to close the gap between our upper and lower bounds on the competitive ratio of randomized algorithms.
Even though several parts of the analysis of the algorithm $\distrand_{\alpha}$ in \cref{thm:distrand} seem ``loose'', we do not know whether the analysis can be improved. 
Observe that the value of $r$ is tight for our proof structure and choice of potential function, as the very last sequence of inequalities in the proof of \cref{lem:supermartingale} is tight. New ideas might be needed to further improve on the analysis (if possible), potentially exploiting the adversary being oblivious.

It is also unclear what the optimal competitive ratio of $\distrand_{\alpha}$ and, more generally randomized cycle exploration is.
The forward-greedy lower bound of $(1+\sqrt{2})/2$ seems like a natural candidate.
However, since (for any choice of $\alpha$) we found small cycles on which the competitive ratio of $\distrand_{\alpha}$ exceeds this value, that would
require a different algorithm.

Indeed, it is unclear whether $\distrand_{\alpha}$ is optimal for the problem.
Finding this algorithm was already a difficult task and we remark that, during the search for this algorithm, we investigated several natural candidates and, for most of them, we found lower bounds of more than 1.366 via a computer search.
Last, we also leave open whether an optimal algorithm for this problem should even be forward-greedy. While the agent can be sure when encountering a state with $b\leq a+d$ that the offline optimum traverses $b$, perhaps adding randomness in such decisions can make the behavior more unpredictable and thus more difficult for the adversary.

\paragraph{Acknowledgments.}
We thank Christian Coester,  Yann Disser, Donald Kougang Yombi, Pascal Schweitzer, and Sergio Tinaharimanjaka for valuable conversations related to this project.
Part of this work was carried out during research visits supported by AvH German Research Chair funding.

\bibliographystyle{alpha}
\bibliography{RandomizedOnlineCycleExploration}

\end{document}